\documentclass[conference]{IEEEtran}
\usepackage[explicit]{titlesec}
\usepackage{graphicx}
\usepackage{easyReview}
\usepackage{verbatim}
\usepackage{epstopdf}
\usepackage{setspace}
\usepackage{mathtools} 
\usepackage{mdframed} 
\usepackage{subfigure} 
\usepackage{algorithmic}
\usepackage{amsmath}
\usepackage{amsthm}
\usepackage{mathrsfs}
\usepackage{extarrows}
\usepackage{cite}
\newtheorem{theorem}{Theorem}

\usepackage{bbm}
\usepackage[bookmarks=false, colorlinks=true, citecolor=blue]{hyperref}
\usepackage{cleveref}
\usepackage{amssymb}

\newcommand{\RNum}[1]{\uppercase\expandafter{\romannumeral #1\relax}}
\newtheorem{Lemma}{Lemma}
\newtheorem{MDP}{MDP}

\newtheorem{problem}{Problem}

\newtheorem{corollary}{Corollary}
\usepackage{stfloats}
\usepackage{color, soul}
\usepackage{booktabs}
\usepackage{tikz,xcolor}
\usetikzlibrary{arrows.meta,automata, positioning,,backgrounds,matrix}

\definecolor{lime}{HTML}{A6CE39}
\usepackage[left=0.68in,right=0.68in,top=0.67in,bottom=0.99in]{geometry}

\titlespacing{\section}{0pt}{1.2ex plus .0ex minus .0ex}{.3ex plus .0ex}
\titlespacing{\subsection}{0pt}{1.2ex plus .0ex minus .0ex}{.3ex plus .0ex}

\makeatletter
\DeclareRobustCommand{\orcidicon}{%
	\begin{tikzpicture}
		\draw[lime, fill=lime] (0,0) 
		circle [radius=0.16] 
		node[white] {{\fontfamily{qag}\selectfont \tiny ID}};    \draw[white, fill=white] (-0.0625,0.095) 
		circle [radius=0.007];    \end{tikzpicture}
	\hspace{-2mm}}
\foreach \x in {A, ..., Z}{%
	\expandafter\xdef\csname orcid\x\endcsname{\noexpand\href{https://orcid.org/\csname orcidauthor\x\endcsname}{\noexpand\orcidicon}}
}

\newcommand*\bigcdot{\mathpalette\bigcdot@{.5}}
\newcommand*\bigcdot@[2]{\mathbin{\vcenter{\hbox{\scalebox{#2}{$\m@th#1\bullet$}}}}}
\makeatother

\usepackage[linesnumbered,ruled,vlined]{algorithm2e}
	
	\IEEEoverridecommandlockouts
	\makeatletter
	\def\@thmnote#1{\textit{[#1]}} 
	\makeatother

\begin{document}
		\title{Optimal Sampling and Scheduling for Remote Fusion Estimation of Correlated Wiener Processes}
		\author{
			Aimin Li and
			Elif Uysal, \emph{Fellow, IEEE}\\
		\IEEEauthorrefmark{1}\textit{Communication Networks Research Group (CNG), EE Dept, METU, Ankara, Turkey} \\
			
		 \textit{E-mail: liaimingogo@gmail.com; uelif@metu.edu.tr}\vspace{-1em}
	\thanks{This work is supported by the European Union through ERC Advanced Grant 101122990-GO SPACE-ERC-2023-A. Views and opinions expressed are, however, those
of the authors only and do not necessarily reflect those of the European Union
or the European Research Council Executive Agency. Neither the European
Union nor the granting authority can be held responsible for them.}	}

		\maketitle
		\allowdisplaybreaks

		\begin{abstract}
			In distributed sensor networks, sensors often observe a dynamic process within overlapping regions. Due to random delays, these \textit{correlated} observations arrive at the fusion center \textit{asynchronously}, raising a central question: \textit{How can one fuse asynchronous yet correlated information for accurate remote fusion estimation?} This paper addresses this challenge by studying the joint design of sampling, scheduling, and estimation policies for monitoring a correlated Wiener process. Though this problem is coupled, we establish a \textit{separation principle} and identify the joint optimal policy: the optimal fusion estimator is a weighted-sum fusion estimator conditioned on Age of Information (AoI), the optimal scheduler is a Maximum Age First (MAF) scheduler that prioritizes the most stale source, and the optimal sampling can be designed given the optimal estimator and the MAF scheduler. To design the optimal sampling,  we show that, under the infinite-horizon average-cost criterion, optimizing AoI is equivalent to optimizing MSE under pull-based communications, \textit{despite the presence of strong inter-sensor correlations}. This structural equivalence allows us to identify the MSE-optimal sampler as one that is AoI-optimal. This result underscores an insight: information freshness can serve as a design surrogate for optimal estimation in \textit{correlated} sensing environments.
		\end{abstract}
		
		\IEEEpeerreviewmaketitle
		
		\section{Introduction}\label{sectionI}
		\subsection{Motivation}
		 	Minimizing Age of Information (AoI) has been a central objective in status update systems, driven by the intuitive notion that \textit{fresher data is more valuable} \cite{kaul2012real}. However, the precise impact of AoI on application-layer performance is not yet fully understood. In single-source settings, prior studies have shown that for tasks such as timely estimation \cite{sun2019wiener,ornee2021sampling,DBLP:journals/ton/TangST24,chen2023sampling,10807024} and inference \cite{shisher2024timely,shisher2022does,shisher2023learning,ari2024goal}, expected task performance can be analytically expressed using AoI under \textit{pull-based communication}. For instance, when observing a Wiener process, the expected mean-square error (MSE) given a message equals its AoI \cite{sun2019wiener}; for Ornstein-Uhlenbeck (OU) processes, the expected MSE is a monotonic function of AoI \cite{ornee2021sampling}. More generally, \cite{shisher2024timely} demonstrates that inference performance may degrade according to a penalty function that is both non-linear and non-monotonic in AoI. These results imply that minimizing AoI or its potentially non-linear penalty can improve the system's estimation/reference performance.

		 	However, in multi-source systems where multiple information streams share a common communication channel, the situation becomes considerably more intricate. New challenges arise. \textit{First}, even simply minimizing the sum of AoI across sources becomes nontrivial, requiring the joint design of sampling and scheduling policies under a shared channel \cite{9358178}. \textit{Second}, the insight into the relationship between AoI and task performance observed in single-source settings does not necessarily hold for multi-source networks. Different sources may have varying importance, dynamics, or update requirements \cite{DBLP:conf/mobihoc/OrneeS23,banawan2023timely}, making it unclear whether minimizing AoI across all sources leads to meaningful application-level gains. \textit{Third}, and more fundamentally, updates from distributed sensors are often \textit{correlated} in practice, \textit{e.g.}, temperature and humidity in environmental monitoring, or video and audio in multimedia systems. In such cases, updates from one source may implicitly convey partial information about others. Consequently, the value of an update is determined not solely by its individual freshness, but by its task-relevant marginal contribution to a certain reference goal \cite{goalaimin,Uysalgoal}. 
		 	\begin{figure}
		 		\centering		 		\includegraphics[width=0.85\linewidth]{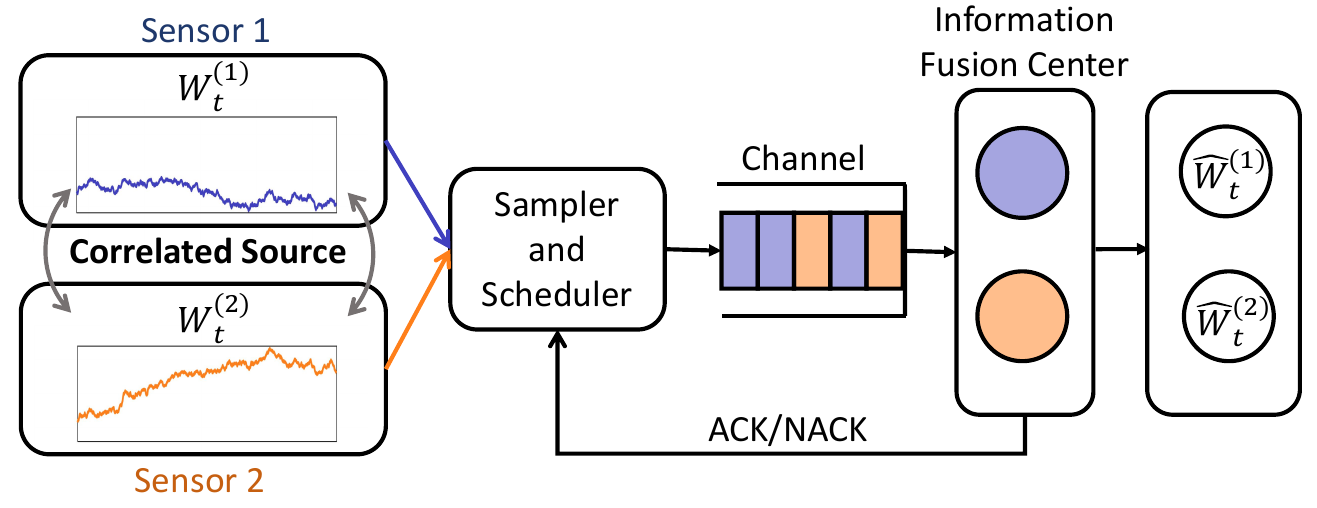}
		 		\caption{System Model.}
		 		\label{fig:}
		 	\end{figure}
		 	\subsection{Related Work}
		 	To address the inter-source correlations in multi-source status update systems, a heuristic metric named \textit{Age of Correlated Information} (AoCI) was proposed in \cite{he2019joint}. Under this model, the AoCI resets only when all relevant correlated updates are successfully received. This metric has been well-studied in various sensor network settings \cite{10304326,zhou2020age,10841476}. However, AoCI overlooks the differential contribution of each source to the final task and imposes an ideal “\textit{all-or-nothing}” assumption that penalizes the system until all relevant updates are received—even when partial observations are already sufficient to capture the key characteristics of the system.
		 	
		 	To address these limitations, some recent works \cite{10978610,hoang2021age,tripathi2022optimizing} propose probabilistic frameworks where the delivery of an update from source $i$ reduces the AoI of source $j$ with probability $p_{ij}$. In \cite{kalor2022timely}, correlated IoT devices are grouped such that an update from one sensor implies concurrent updates in the same group. Additionally, \cite{tong2022age} introduces an attenuation factor $\alpha_n(t)$ to model inter-source correlation, enabling partial AoI reduction based on the statistical overlap between sources. 
		 	
		 	While the aforementioned models capture information overlap caused by inter-source correlation, they often treat freshness as an abstract metric, largely decoupled from the system’s estimation and other inference objectives. Consequently, although heuristic age metrics and AoI dynamics are introduced, their direct impact on application-layer estimation accuracy remains insufficiently explored.
		 	
		 	 A notable exception that explicitly links information freshness with estimation accuracy for correlated sources is the recent work by Ramakanth et al.~\cite{10700668}, which addresses joint scheduling and estimation for monitoring a correlated Wiener process. Their framework, built on a time-slotted system with \textit{idealized}, \textit{zero-delay} communication channel, offers important insights into how information freshness can help improve the estimation error of correlated sources. However, real-world networks are characterized by \textit{unavoidable} random delays, due to factors like buffer overflow, link-layer retransmissions, and routing variability—especially in dynamic environments such as Non-Terrestrial Networks (NTNs). Yet, in scenarios with random transmission delays \cite{sun2017update,DBLP:journals/ton/PanBSS23,li2024sampling}, this sampling decision becomes far from trivial: the \textit{zero-wait} sampling, while simple, can severely degrade estimation performance.

			\subsection{Contributions of This Work} 
			\begin{itemize}
				\item \textbf{System Model}: We investigate the joint design of sampling, scheduling, and estimation policies for \textit{correlated} Wiener processes in the presence of random delays. Unlike the time-slotted model in \cite{10700668}, which considers ideal immediate delivery and focuses only on scheduling, our continuous-time framework incorporates stochastic delays typical in real-world networks. Moreover, we include sampling decisions as part of the design, which determines when each sensor should generate and transmit updates. This is motivated by recent findings showing that non-zero-waiting policies can outperform zero-waiting policies under random delay \cite{sun2017update,DBLP:journals/ton/PanBSS23,li2024sampling}.
				
				\item\textbf{Separation-Structured Optimal Design:} We establish a separation principle for the joint design of sampling, scheduling, and estimation in correlated sensing. This principle enables a modular decomposition of the problem, leading to a structured policy that achieves optimal long-term average MSE in estimation and scheduling, and near-optimal performance in sampling. The joint policy comprises: (i) a weighted-sum fusion estimator; (ii) a Maximum Age First (MAF) scheduler; and (iii) a threshold-triggered Water-Filling (WF) sampler. 
				
				\item\textbf{Methodology Design:} Solving the joint optimization problem poses a key technical challenge, as the state space of the formulated infinite-horizon MDP is both \textit{transient} and \textit{unbounded}, rendering classical average-cost approaches such as the \textit{Average Cost Optimality Equation} (ACOE) and \textit{Relative Value Iteration} (RVI) inapplicable. Existing remedies, including \textit{discounting} and \textit{state truncation} \cite{leng2019age,10778570}, often introduce bias or compromise optimality. In contrast, we construct a new MDP with \textit{bounded} and \textit{recurrent} dynamics. We prove that the transformed MDP is equivalent to the original MDP in terms of average-cost performance. This reduction allows us to shift the analysis from an intractable MDP to a well-behaved MDP, thereby enabling efficient optimization while preserving theoretical optimality. 
				
				\item\textbf{Theoretical Insight:} We reveal that the transformed well-behaved MDP is mathematically equivalent to a classical long-term average sum-AoI minimization problem. This equivalence implies that, under the infinite-horizon average-cost criterion, an AoI-optimal sampling policy for multiple sources \cite{9358178} also minimizes the MSE, \textit{even in the presence of strong inter-source correlations}.
				
			\end{itemize}
			\section{System Model and Problem Formulation }\label{section III}

            
            We consider a system where spatially distributed sensors monitor a dynamic physical process (e.g., temperature or pressure), with correlated observations due to overlapping sensing regions. Sensor samples are transmitted over a shared random-delay channel to a remote estimator, which fuses the collected information and reconstructs the underlying process in real-time. The system model is shown in Fig. \ref{fig:}. 
			
			\subsection{Coupled Correlated Sources}
			The monitored dynamic source is modeled as a correlated bivariate Wiener process \( \mathbf{W}_t = (W_t^{(1)}, W_t^{(2)})^\top \), constructed from two independent standard Wiener processes \( B_t^{(1)} \) and \( B_t^{(2)} \), defined on a common filtered probability space \( (\Omega, \mathcal{F}, \{\mathcal{F}_t\}_{t \ge 0}, \mathbb{P}) \). Specifically, the dynamics are given by:
			\begin{equation}
				\label{eq:Wiener-dynamics}
				\begin{aligned}
					dW_t^{(1)} &= dB_t^{(1)}, \\
					dW_t^{(2)} &= \rho \, dB_t^{(1)} + \sqrt{1 - \rho^2} \, dB_t^{(2)},
				\end{aligned}
			\end{equation}
			where \( \rho \in [0, 1] \) quantifies the correlation due to distributed sensors' spatial deployment or sensing modalities.
			
			
			
			\subsection{Multi-source Network Model}
			Consider two distributed sensors where the sensor $m\in\{1,2\}$ monitors $W_t^{(m)}$. Let \( S_i \) denote the sampling time of the \( i \)-th sample, and let \( a_i \in \{1, 2\} \) indicate the scheduled sensor corresponding to that sample. Each sample \( W_{S_i}^{(a_i)} \) is encapsulated into a packet \( (S_i, W_{S_i}^{(a_i)}) \) and transmitted to a remote data center for remote information fusion estimation. Each packet experiences a random transmission delay \( Y_i \in \mathcal{Y} \) due to factors such as network handover, congestion, packet size variability, and retransmissions. The delay support \( \mathcal{Y} \subset \mathbb{R}^+ \) is assumed to be finite and bounded. We adopt a non-preemptive scheduling policy: a new sample cannot be generated while the communication channel is occupied by an ongoing transmission. Let \( D_i = S_i + Y_i \) denote the delivery time of the \( i \)-th sample. To enhance remote estimation performance, the sampler is allowed to insert a deliberate waiting time \( Z_i \in \mathcal{Z} \) before generating the next sample\footnote{As discussed in \cite{sun2019wiener}, sampling immediately upon delivery (i.e., \textit{zero-wait} sampling), although \textit{throughput-optimal}, is suboptimal for remote estimation.}, where \( S_{i+1} = D_{i} + Z_i \). The waiting time set \( \mathcal{Z} \subset \mathbb{R}^+ \) is also assumed to be finite and bounded. Without loss of generality, we assume $D_0=0$ and $S_0=0$.
			
			At any time $t$, the most recently delivered sample from source $m$ is generated at time:
			\begin{equation}
				U_m(t)=\max\{S_i:a_i=m,D_i\le t\},\quad \forall m\in\{1,2\}.
			\end{equation}
			Accordingly, the AoI of source \( m \) at time \( t \) is defined as:
			\begin{equation}
				\Delta_m(t)\triangleq t-U_m(t),\forall m\in\{1,2\}.
			\end{equation}
			
		\subsection{Sampling, Scheduling, and Estimating Policies}	
		We consider a \textit{pull-based} communication framework in which all decisions, including sampling, scheduling, and fusion-based estimation, are made by the information fusion center. These decisions are made \textit{causally}, relying on past sampling and scheduling actions, as well as previously received sample updates. At any time \( t \in \mathbb{R}^+ \), the information available for making such \textit{causal} decisions(\textit{history}) is given by:
		\begin{equation}
			\mathcal{I}_t = \left\{ \left(a_i, S_i, W_{S_i}^{(a_i)},D_i \right) \,\middle|\, D_i \le t \right\}.
		\end{equation}	
		 Here, \( a_i \in \{1,2\} \) indicates the source index, \( S_i \) is the sampling time, \( D_i \) is the delivery time, and \( W_{S_i}^{(a_i)} \) is the corresponding sample value.  
		 
		 A \textit{causal policy} at the fusion center includes the following components: 
		\begin{itemize}
			\item \textbf{Estimator} \( g \): This function determines the real-time estimate \( \hat{\mathbf{W}}_t \) based on the \textit{history} \( \mathcal{I}_{t} \):
			\begin{equation}
				\hat{\mathbf{W}}_t = (\hat{W}_t^{(1)}, \hat{W}_t^{(2)})^\top = g(\mathcal{I}_t), \quad \forall t \in \mathbb{R}^+.
			\end{equation}
			
			\item \textbf{Scheduler} \( \pi \): A scheduler is a policy function that determines which source is selected for transmission. The selected source during the $i$-th transmission is based on the \textit{history} \( \mathcal{I}_{D_{i-1}} \):
			\begin{equation}
				a_i=\pi(\mathcal{I}_{D_{i-1}}),\quad\forall i\in\mathbb{N}.
			\end{equation}
			This policy leads to a scheduling order $(a_1,a_2,\cdots)$.
			
			\item \textbf{Sampler} \( f \): 
			A sampler is a function that determines when each sensor is sampled. The sampling time of the $i$-th packet is determined based on the \textit{history} \( \mathcal{I}_{D_{i-1}} \):
			\begin{equation}
			S_i = f(\mathcal{I}_{D_{i-1}}),\quad\forall i\in\mathbb{N}.
			\end{equation}
			Equivalently, the sampler can be defined via the waiting time after the previous delivery, where \( Z_i = S_i - D_{i-1}\) and $Z_i\in\mathcal{Z} $, leading to the sequence \( (Z_1, Z_2, \dots) \).
		\end{itemize}			
		
	\subsection{Problem Formulation}
	We aim to design sampling, scheduling, and estimation policies that jointly minimize the long-term average MSE. For any admissible policy triple \( (g, \pi, f) \), the performance is evaluated via the long-term average MSE:
		\begin{equation}
			\overline{\mathrm{MSE}}(g, \pi, f) \triangleq \limsup_{T \to \infty}\frac{1}{T} \, \mathbb{E}_{g,\pi,f} \left[ \int_0^T \left\| \mathbf{W}_t - \hat{\mathbf{W}}_t \right\|_2^2 \, \mathrm{d}t \right],
		\end{equation}
		where \( \mathbf{W}_t \in \mathbb{R}^2 \) is the source process and \( \hat{\mathbf{W}}_t\in\mathbb{R}^2 \) is its remote fusion estimate. 
		In this work, we focus on the following optimization problem:

			\begin{problem}\label{p1}
			Find a policy triple \( (g, \pi, f) \) that minimizes the long-term average MSE:
				\begin{equation}\label{MSEltoptimal}
					\varphi^{\star} \triangleq \inf_{g,\pi,f} \, \overline{\mathrm{MSE}}(g,\pi,f),
				\end{equation}
				where \(\varphi^{\star}\) denotes the optimal long-term average MSE.
				\end{problem}
		
		\section{Optimal Estimating, Sampling, and Scheduling}\label{Section3}
		\subsection{Separation Principle and Optimal Estimator}
			Under a given scheduling policy \( \pi \) and a sampling policy \( f \), the minimum mean-square error (MMSE) estimator \cite[IV.B.8]{poor2013introduction}, which minimizes the instantaneous mean-square error in the system, is the conditional expectation of each source process given the available \textit{history}:
		\begin{equation}
			\begin{aligned}				\hat{W}_t^{(1)}=\mathbb{E}\left[{W}_t^{(1)}\mid\mathcal{I}_t\right], \quad	\hat{W}_t^{(2)}=\mathbb{E}\left[{W}_t^{(2)}\mid\mathcal{I}_t\right].		
			\end{aligned}
		\end{equation}
		The following theorem provides an explicit characterization of the MMSE estimator.

		\begin{theorem}[\textit{MMSE Fusion Estimator for Correlated Wiener Processes}]\label{the1}
			Let \( \rho \in [0,1] \) denote the correlation coefficient between \( W_t^{(1)} \) and \( W_t^{(2)} \). Let $\Delta_{m}(t)$ denote the AoI of source~$m$. Given any scheduling policy $\pi$ and sampling policy $f$, the MMSE estimate of $W_t^{(1)} $ is:
			\begin{equation}\label{eq11}
				\hat{W}_t^{(1)} =
				\begin{cases}
					W^{(1)}_{t - \Delta_1(t)}, &
					\Delta_1(t) \le \Delta_2(t) \\
					\begin{aligned}
						&G_{t,\rho}(\Delta_1(t), \Delta_2(t)) \cdot W^{(1)}_{t - \Delta_1(t)} + \\
						&Q_{t,\rho}(\Delta_1(t), \Delta_2(t)) \cdot W^{(2)}_{t - \Delta_2(t)}
					\end{aligned}, &
					\vcenter{\hbox{$\Delta_1(t) > \Delta_2(t)$}}
				\end{cases}
			\end{equation}	
			The MMSE estimation of $W_t^{(2)} $ is:
			\begin{equation}\label{eq12}
				\hat{W}_t^{(2)} =
				\begin{cases}
					W^{(2)}_{t - \Delta_2(t)}, &
					\Delta_2(t) \le \Delta_1(t) \\
					\begin{aligned}
						&G_{t,\rho}(\Delta_2(t), \Delta_1(t)) \cdot W^{(2)}_{t - \Delta_2(t)} + \\
						&Q_{t,\rho}(\Delta_2(t), \Delta_1(t)) \cdot W^{(1)}_{t - \Delta_1(t)}
					\end{aligned}, &
					\vcenter{\hbox{$\Delta_2(t) > \Delta_1(t)$}}
				\end{cases}
			\end{equation}
			where the functions $G_{t,\rho}(x,y)$ and $Q_{t,\rho}(x,y)$ are defined by:
			\begin{equation}
			\resizebox{1\hsize}{!}{$	\begin{aligned}			G_{t,\rho}(x,y)\triangleq\frac{(1 - \rho^2)(t - y)}{t - y - \rho^2(t - x)},\quad
					Q_{t,\rho}(x,y)\triangleq\frac{\rho(x - y)}{t - y - \rho^2 (t - x)},
				\end{aligned}$}
			\end{equation}
			and we refer to them as {fusion weights}.
		\end{theorem}
		
		\begin{proof}
			See Appendix \ref{proof:the1}.
		\end{proof}		
		The results above reveal how the remote data fusion center exploits asynchronous and potentially \textit{stale} information from distributed sensors to perform real-time estimation of correlated sources. Specifically, the MMSE estimator for \( W_t^{(m)} \) depends on the relative freshness (i.e., AoI) of the most recently received samples from both sources. If the latest sample from source \( m \) is fresher than the other source $m'$, i.e., \( \Delta_m(t) \le \Delta_{m'}(t) \), then the MMSE estimate of \( W_t^{(m)} \) is simply the value of its most recently received sample:
		\begin{equation}
		\hat{W}_t^{(m)} = W_{t - \Delta_m(t)}^{(m)}.
		\end{equation}
		This corresponds to the intuitive case in which fresher information dominates. However, if the sample from the source $m'$ is fresher, i.e., \( \Delta_m(t) > \Delta_{m'}(t) \), then the MMSE estimate of \( W_t^{(m)} \) must fuse both the staler sample \( W^{(m)}_{t - \Delta_m(t)} \) and the fresher correlated sample \( W^{(m')}_{t - \Delta_{m'}(t)} \). These two components are weighted by coefficients that depend on the correlation \( \rho \), the AoIs $\Delta_1(t)$ and $\Delta_2(t)$, and the current time $t$. Moreover, by choosing specific values of $\rho$, our results recover known estimation structures studied in the existing literature.
		
		\begin{itemize}
			\item If \( \rho = 1 \), the fusion coefficients simplify to \( G_{t,1}(x, y) \equiv 0 \) and \( Q_{t,1}(x, y) \equiv 1 \). Substituting these into the MMSE estimator (Theorem \ref{the1}) leads to the degenerate result: \begin{equation}
				\hat{W}_t^{(1)}=\hat{W}_t^{(2)}=\begin{cases}
					{W}_{t-\Delta_1(t)}^{(1)},&\text{if } \Delta_1(t)\le\Delta_2(t)
					\\{W}_{t-\Delta_2(t)}^{(2)},&\text{if } \Delta_1(t)>\Delta_2(t).
				\end{cases}
			\end{equation}
			Note that $W_t^{(1)}=W_t^{(2)}$ when $\rho=1$; in this case, the problem reduces to the classical single-source remote estimation scenario, as studied in \cite{sun2019wiener}.
			
			\item If \( \rho = 0 \), the fusion weights reduce to \( G_{t,0}(x, y) \equiv 1 \) and \( Q_{t,0}(x, y) \equiv 0 \). Substituting these into the MMSE estimator yields: \begin{equation}
				\hat{W}_t^{(m)} = W_{t - \Delta_m(t)}^{(m)}, \quad \forall m.
			\end{equation}
			This corresponds to the classical setting for remote estimation of \textit{independent} sources, where each source can be estimated independently, as discussed in~\cite{DBLP:conf/mobihoc/OrneeS23} and~\cite{banawan2023timely}.
		\end{itemize}
		\subsection{Maximum Age First Scheduling is All You Need}
		
		In pull-based communication systems, the receiver pulls updates based on its own information requirements. Given the available information \textit{history} $\mathcal{I}_t$, the expected MSE at time $t$ is given by:
		\begin{equation}\label{Delta}
			\begin{aligned}
				&\epsilon(t) \triangleq \mathbb{E}\left[\|\mathbf{W}_t - g_{\mathrm{MMSE}}(\mathcal{I}_t)\|_2^2 \,\middle|\, \mathcal{I}_t \right]
				= \Delta_1(t) + \Delta_2(t) \\
				&- \frac{\rho^2 \left( \Delta_1(t) - \Delta_2(t) \right)^2}
				{(1 - \rho^2)t + \rho^2 \max\{\Delta_1(t), \Delta_2(t)\} - \min\{\Delta_1(t), \Delta_2(t)\}}.
			\end{aligned}
		\end{equation}
		This result reveals a significant insight: \textit{the estimation error depends not only on the individual values of $\Delta_1(t)$ and $\Delta_2(t)$, but also their interaction adjusted by $\rho$.} When $\rho = 0$, the sources are independent, and the estimation error simplifies neatly to the sum of AoI values:
		\begin{equation}\label{c}
			\epsilon(t) = \Delta_1(t) + \Delta_2(t),
		\end{equation}
		a form that aligns with prior results in uncorrelated multi-source settings, \textit{e.g.},  \cite{DBLP:conf/mobihoc/OrneeS23} and \cite{banawan2023timely}. Consequently, optimizing update scheduling in this setting is equivalent to minimizing the sum of AoI \cite{9358178}. 
		
		In this work, we deliberately focus on the more complex and practical regime in which the correlation between sources is strictly positive ($\rho>0$). Let $S_{i,m} = \max\{S_k : a_k = m,\, S_k < D_i\}$ denote the most recent sampling time from source $m$ prior to the $i$-th update interval. Since only one source is sampled per interval, the AoI for source $m$ evolves linearly during $[D_i, D_{i+1})$ as:
		\begin{equation} \label{delta}
			\Delta_m(t) = t - S_{i,m}, \quad \text{for } t \in [D_i, D_{i+1}).
		\end{equation}	
		Substituting~\eqref{delta} into \eqref{Delta}, the instantaneous estimation error $\epsilon(t)$ can be expressed explicitly as a function of $S_{i,1}$ and $S_{i,2}$:
		\begin{equation}
			\begin{aligned}
					\epsilon(t)& = (t - S_{i,1}) + (t - S_{i,2}) - \\&\frac{\rho^2 (S_{i,1} - S_{i,2})^2}{\max\{S_{i,1}, S_{i,2}\} - \rho^2 \min\{S_{i,1}, S_{i,2}\}}, \quad \forall t \in [D_i, D_{i+1}).
			\end{aligned}
		\end{equation}
		The next step is to evaluate the cumulative estimation error accumulated over each update interval $[D_i, D_{i+1}) $. Let the length of each update interval be $D_{i+1}-D_i = Y_{i+1}+Z_i$. Let $\mu_Y=\mathbb{E}[Y]$ and $\sigma_Y=\mathbb{E}[Y^2]$ denote the first and second moments of the delay distribution. Then, the expected accumulated MSE within the interval $[D_i,D_{i+1})$ takes the following form:
		\begin{equation}\label{eq21}
			\begin{aligned}
				&\mathbb{E}_{}\left[\left.\int_{D_{i}}^{D_{i}+Y_{i+1}+Z_{i}}\epsilon(t)dt\right|\mathcal{I}_{D_{i}}\right]\\&=\sigma_Y+\mu_YZ_i+\left(Z_{i}+\mu_Y\right)\left(2Y_{i}+Z_i+q_{\rho}\left(S_{i,1},S_{i,2}\right)\right).
			\end{aligned}
		\end{equation}
		Here, the function $q_{\rho}(x,y)$ is defined as:
		\begin{equation}
			\begin{aligned}					q_{\rho}(x,y)=\left|x-y\right|\times\left(1-\frac{\rho^2|x-y|}{\max\{x,y\}-\rho^2\min\{x,y\}}\right).
			\end{aligned}
		\end{equation}
		For notational simplicity, let $\mathcal{C}_{\rho}(S_{i,1},S_{i,2},Y_i,Z_i)$ denote the right-hand side of \eqref{eq21},  representing the total expected cost incurred during the $i$-th interval. Building on this definition, the long-term average MSE minimization problem can be reformulated as minimizing the long-term time-average of per-interval expected costs:
		\begin{align}\label{eq23}
			&\min_{\pi,f}\lim\limits_{T\to\infty}\frac{1}{T} \, \mathbb{E} \left[ \int_0^T \left\| \mathbf{W}_t - \hat{\mathbf{W}}_t \right\|_2^2 \, \mathrm{d}t \right]\notag\\
			&=\min_{\pi,f}\lim\limits_{n\to\infty}\frac{\sum_{i=0}^{n-1}\mathbb{E}\left[\mathbb{E}\left[\left.\int_{D_{i}}^{D_{i+1}}\epsilon(t)dt\right|\mathcal{I}_{D_{i}}\right]\right]}{\mathbb{E}[D_n]}\notag\\
			&=\min_{\pi,f}\lim_{ n \to \infty } \frac{\frac{1}{n}\sum_{i=0}^{n-1}\mathbb{E}\left[\mathcal{C}_{\rho}(S_{i,1},S_{i,2},Y_i,Z_i) \right]}{\frac{1}{n}\sum_{i=0}^{n-1}\mathbb{E}[Y_{i+1}+Z_i]}.
		\end{align}
		To address this fractional optimization, we adopt the classical Dinkelbach method~\cite{dinkelbach1967nonlinear}. For any $\lambda\ge0$, , we define the following parameterized auxiliary problem:
		\begin{problem}[\textit{Average–cost reformulation via Dinkelbach}]\label{p3}
			\begin{equation}\label{eq:Dinkelbach-Theta}
				\resizebox{1\hsize}{!}{$\begin{aligned}
					\Theta(\lambda)=
					\min_{\pi,f}
					\lim_{n\to\infty}
					\frac{1}{n}\sum_{i=0}^{n-1}
					\mathbb{E}\!\Bigl[
					\mathcal{C}_{\rho}\bigl(S_{i,1},S_{i,2},Y_i,Z_i\bigr)
					-\lambda\bigl(Y_{i+1}+Z_i\bigr)
					\Bigr].
				\end{aligned}$}
			\end{equation}
		\end{problem}		
		The following lemma establishes the equivalence between the fractional problem \eqref{eq23} and its reformulation \eqref{eq:Dinkelbach-Theta}.
		
		\begin{Lemma}[\textit{Equivalence via Dinkelbach}]\label{lem:Dinkelbach}
			Let \(\varphi^\star\) denote the optimal value of the fractional problem in~\eqref{eq23}.  
			Then the following assertions hold: 
			\begin{enumerate}
				\item $\varphi^{\star} \gtreqless \lambda \text { if and only if } \Theta(\lambda) \gtreqless 0 \text {. }$
				\item When $\Theta(\lambda)=0$, then any policy that minimizes~\eqref{eq:Dinkelbach-Theta} also solves the original fractional problem~\eqref{eq23}.
			\end{enumerate}
		\end{Lemma}
		\begin{proof}
			The proof follows standard arguments in fractional programming and is omitted here due to page constraints.
		\end{proof}
		Problem \ref{p3} admits a natural \emph{infinite–horizon average–cost MDP} formulation:
		\begin{MDP}[\textit{Original Infinite-Horizon MDP}]\label{mdp1} This MDP can be specified by the following components:
			\begin{itemize}
				\item \textbf{State}:  
				\(x_i\triangleq\bigl(S_{i,1},S_{i,2},Y_i\bigr)\in\mathbb{R}_+^{3}\).
				\item \textbf{Action}:  
				\(u_i\triangleq(a_i,Z_i)\) with scheduling action \(a_i\in\{1,2\}\) and  waiting time \(Z_i\in\mathbb{R}_+\).
				\item \textbf{One–step cost}:  
				\begin{equation}
					\begin{aligned}
						c_{\lambda}^{\rho}(x_i,u_i)&\triangleq
						\mathcal{C}_{\rho}\bigl(S_{i,1},S_{i,2},Y_i,Z_i\bigr)
						-\lambda\bigl(\mathbb{E}[Y_{i+1}]+Z_i\bigr)\\&=\sigma_Y+\mu_YZ_i-\lambda\mu_Y-\lambda Z_i\\&+\left(Z_{i}+\mu_Y\right)\left(2Y_{i}+Z_i+q_{\rho}\left(S_{i,1},S_{i,2}\right)\right).
					\end{aligned}
				\end{equation}
				which coincides with the right–hand side of~\eqref{eq:Dinkelbach-Theta}.
				\item \textbf{State transition}:  
				Conditioned on \((x_i,u_i)\), for any source index $m$, the sampling time evolves as:
				\begin{equation}
					S_{i+1,m}=
					\begin{cases}
						\displaystyle
						\max\{S_{i,1},S_{i,2}\}+Y_i+Z_i, &\text{if } a_i=m,\\[4pt]
						S_{i,m}, &\text{otherwise},
					\end{cases}
				\end{equation}
				while \(Y_{i+1}\) is drawn i.i.d. from the prescribed delay distribution and is independent of \(Y_i\).
			\end{itemize}
		\end{MDP}

			However, in the MDP above, the state $x_i$ may be \emph{transient and unbounded}, with the sampling times \(S_{i,1},S_{i,2}\) growing without bound. As a result, standard dynamic programming tools for average-cost problems, such as the ACOE \cite[Eq. 4.1]{howard1960dynamic} and its solution via RVI\cite{white1963dynamic} become inapplicable. To address this technical challenge, we propose a two-stage MDP transformation. We first introduce the following auxiliary variables:
			\begin{equation}\label{eq:Gamma-M}
					\Gamma_i \;=\;
					|S_{i,1}-S_{i,2}|,\quad
					M_i      \;=\;
					\max\{S_{i,1},S_{i,2}\},\quad\forall i
			\end{equation}
			where $\Gamma_i$ denotes the age gap between the two sources and $M_i$ captures their common upper envelope. We redefine the state of the system as the pair $(\Gamma_i, M_i)$. Based on this state space transformation, we construct the following equivalent MDP:
			\begin{MDP}[\textit{State Space Transformation}]\label{mdp2} This MDP is defined by the following components:			
								\begin{itemize}
					\item \textbf{State}:  
					\(x_i\triangleq\bigl(\Gamma_i, M_i,Y_i\bigr)\in\mathbb{R}_+^{3}\).
					\item \textbf{Action}:  
					\(u_i\triangleq(a_i,Z_i)\in\{1,2\}\times\mathbb{R}_+\).
					\item \textbf{One–step cost}:  
					\begin{equation}\label{eq:one-step-cost}
						\begin{aligned}
							&c_{\lambda}^{\rho}(x_i,u_i)=\sigma_Y+\mu_YZ_i-\lambda\mu_Y-\lambda Z_i\\&+\left(Z_{i}+\mu_Y\right)\left(2Y_{i}+Z_i+H_{\rho}\left(\Gamma_i,M_i\right)\right),
						\end{aligned}
					\end{equation}
					where $H_{\rho}(\Gamma_i,M_i)
					\;=\;q_{\rho}(S_{i,1},S_{i,2})$, with
					\begin{equation}
						H_{\rho}(\Gamma_i,M_i)
						\;=\;
						\Gamma_i\!
						\left(
						1-\frac{\rho^{2}\Gamma_i}{\,\Gamma_i+(1-\rho^{2})M_i\,}
						\right).
					\end{equation}
					
					\item \textbf{State transition}:  
					$M_i$ evolves as:
					\begin{equation}\label{Mdynamics}
						M_{i+1} = M_i + Y_i + Z_i.
					\end{equation}
					$\Gamma_i$ evolves as:
					\begin{equation}\label{eq:gap-env-transition}
						\begin{aligned}
							\Gamma_{i+1} =
							\begin{cases}
								\Gamma_i + Y_i + Z_i,
								& \text{if } a_i=\mathop{\arg\max}_{m}\{S_{i,m}\},\\[4pt]
								Y_i + Z_i,
								& \text{if } a_i=\mathop{\arg\min}_{m}\{S_{i,m}\},
							\end{cases}
						\end{aligned}
					\end{equation}
				\end{itemize}
			\end{MDP}
	With this MDP, although \(M_i\) may still grow unbounded over time, \(\Gamma_i\) can be bounded under a suitable policy. This motivates us to examine a structural property of the one-step cost function, captured in the following lemma.
			\begin{Lemma}[\textit{Monotonicity in the age gap \(\Gamma\)}]\label{lemma2}
				The function \(H_{\rho}(M,\Gamma)\) as well as the one–step cost
				\(c_\lambda^{\rho}(\Gamma,M,Y,u)\) in~\eqref{eq:one-step-cost} are strictly increasing
				in \(\Gamma\). 
			\end{Lemma} 
			\begin{proof}
				For every fixed envelope level \(M\) we have
				\begin{equation}
					\begin{aligned}
						\frac{\partial H_{\rho}(M,\Gamma)}{\partial \Gamma}=1-\frac{\rho^2\Gamma(\Gamma+2(1-\rho^2)M)}{(\Gamma+(1-\rho^2)M)^2}>1-\rho^2>0,
					\end{aligned}
				\end{equation}
				Hence, \(H_{\rho}(M,\Gamma)\) is strictly increasing in \(\Gamma\). It follows immediately that \(c_\lambda^{\rho}(\Gamma,M,Y,u)\) is also  increasing with $\Gamma$.
			\end{proof}		
			Leveraging the monotonicity property in Lemma~\ref{lemma2}, we now show that the MAF scheduling policy is optimal for the transformed MDP \ref{mdp2}.
			\begin{theorem}[\textit{Separation principle and optimality of MAF scheduling}]\label{thm:MAF}
				For every \(\lambda\ge0\) and any given sampling policy~$f$, the stationary policy that
				\emph{always} schedules the sensor with the maximum age,
				\begin{equation}\label{eq33}
				a_i^{\textsf{MAF}}
				=\mathop{\arg\max}_{m\in\{1,2\}}\Delta_{m}(D_{i}),\quad\forall i\in\mathbb{N},
				\end{equation}
				is optimal for the average–cost MDP defined by
				\eqref{eq:one-step-cost}--\eqref{eq:gap-env-transition}.
			\end{theorem}
			\begin{proof}
				Fix an arbitrary epoch index \(\ell\).  Construct two parallel action sequences:		
				$
				\mathbf u=((a_1,Z_1),\dots,(a_\ell,Z_\ell),\dots),
				\mathbf u'=((a_1',Z_1'),\dots,(a_\ell',Z_\ell'),\dots)$. Here $Z_k'=Z_k$, for $\forall k$. The scheduling action $a_\ell$ is:
				\begin{equation}
					a_\ell=\mathop{\arg\min}_{m\in\{1,2\}}\Delta_{m}(D_{\ell}),
				\end{equation}
				while the action $a_k'$ only differs $a_k$ at the epoch $\ell$\footnote{As $\Delta_m(D_{\ell})=D_\ell-S_{\ell,m}$, minimizing(maximizing) $\Delta_m(D_\ell)$ is equivalent to maximizing(minimizing) $S_{\ell,m}$.}:
					\begin{equation} 
						a_k'=\begin{cases}
							a_k&\text{if }k\ne \ell,\\
							\mathop{\arg\max}_{m}\Delta_{m}(D_\ell)&\text{if }k=\ell,
						\end{cases}
				\end{equation}
			Let the exogenous random delay \(\{Y_k\}_{k\ge 1}\) be the
				same for both systems.  Denote the resulting state trajectories by
				\(\{\Gamma_k,M_k\}_{k\ge1}\) and \(\{\Gamma_k',M_k'\}_{k\ge1}\), respectively. Because \(M_k\) evolves regardless of which sensor is scheduled,
				and the two sequences share the same \(\{Y_k,Z_k\}\), from \eqref{Mdynamics} we have
				\begin{equation}\label{Mcompare}
					M_k = M_k',\qquad\forall k\in\mathbb{N}.
				\end{equation}
				Then, we apply \textit{induction} in Appendix \ref{appendixa} to establish:
				\begin{equation}\label{Gammacompare}
					\Gamma_{k}'\le\Gamma_k, \qquad\forall k\in\mathbb{N}.
				\end{equation}
				Given the strict monotonicity of the cost function $c_\lambda^{\rho}(\Gamma,M,Y,u)$ in $\Gamma$ (\ref{lemma2}), it immediately follows that:
				\begin{equation}\label{compareonestep}
					c_\lambda^{\rho}(\Gamma_k',M_k',Y_k,(a_k',Z_k'))\le c_\lambda^{\rho}(\Gamma_k,M_k,Y_k,(a_k,Z_k)),\quad\forall k.
				\end{equation}
				Consequently, deviating from the MAF policy at any epoch results in a higher total cost, confirming that consistently scheduling the sensor with the maximum age is optimal.  
			\end{proof}	
			Theorem~\ref{thm:MAF} reveals a key simplification: \textit{scheduling and sampling can be decoupled}, with the optimal scheduling given by the intuitive MAF rule. Consequently, joint policy optimization reduces to designing the sampling policy under fixed $g_{\mathrm{MMSE}}$ and $\pi_{\mathrm{MAF}}$,  as addressed in the next subsection.

		\subsection{Optimal Sampling Policy}
		Given the MAF scheduling, the \textit{age gap} dynamics in \eqref{eq:gap-env-transition} becomes \textit{recurrent}: $\Gamma_{i+1}=Y_i+Z_i$. However, the state $M_i$ remains \textit{transient} and \textit{unbounded}. To address this, we here eliminate the \textit{tricky} state $M_i$ and formulate a new MDP:
		\begin{MDP}[\textit{$M_i$ State Elimination Given the MAF Scheduling}]\label{mdp3} This MDP is expressed by four components:			
			\begin{itemize}
				\item \textbf{State}:  
				\(x_i\triangleq\bigl(\Gamma_i,Y_i\bigr)\in\mathbb{R}_+^{2}\).
				\item \textbf{Action}:  
				\(Z_i\in\mathbb{R}_+\).
				\item \textbf{One–step cost}:  
				\begin{equation}\label{eq:one-step-cost-compare}
					\begin{aligned}
						\psi_\lambda(\Gamma_i,&Y_i,u_i)=\sigma_Y+\mu_YZ_i-\lambda\mu_Y-\lambda Z_i\\&+\left(Z_{i}+\mu_Y\right)\left(2Y_{i}+Z_i+\Gamma_i\right).
					\end{aligned}
				\end{equation}
				
				\item \textbf{State transition}:  
				$\Gamma_i$ evolves as:
				\begin{equation}\label{dynamics2Gamma}
					\Gamma_{i+1} =Y_i + Z_i,
				\end{equation}
				while \(Y_{i+1}\) is independent of \(Y_i\).
			\end{itemize}
		\end{MDP}
		
		The following theorem establishes that, under the \textit{average cost criterion}, it suffices to solve MDP \ref{mdp3} rather than MDP \ref{mdp2} when the MAF scheduling policy is employed.
		\begin{theorem}[\textit{Average-cost MDP Equivalence}]\label{lemma3}
			Given the MAF scheduling policy, for any sampling policy \((Z_0,Z_1,\ldots)\), the long-term average costs corresponding to \(c_\lambda^{\rho}(\cdot)\) in MDP \ref{mdp2} and \(\psi_\lambda(\cdot)\) in MDP \ref{mdp3} coincide exactly. Formally,  
			\begin{equation}\label{equivalence}
				\lim_{n\to\infty}
				\frac{1}{n}\sum_{i=0}^{n-1}
				\mathbb{E}[\psi_\lambda(\Gamma_i,Y_i,u_i)-c_\lambda^{\rho}(\Gamma_i,M_i,Y_i,u_i)] = 0.
			\end{equation}
		\end{theorem}
		\begin{proof}
			See Appendix \ref{proof:the3}.
		\end{proof}
		
		With this transformation, the previously intractable MDP with \textit{transient} and \textit{unbounded} states becomes analytically tractable. Since \(\Gamma_i\) is \textit{recurrent}, standard solution methods, such as the RVI algorithm, can be applied directly to solve this problem.
		
		\subsection{Age-Optimal Sampling is All you Need}
		\begin{figure}[h]
			\centering
			\includegraphics[width=0.7\linewidth]{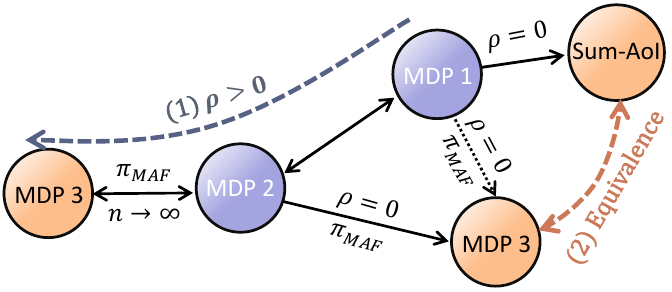}
			\caption{Relationship between different MDP formulations.}
			\label{fig:simulation2}
		\end{figure}
		An insightful special case arises when $\rho=0$. In this scenario, the function $H_{0}(M,\Gamma)$ reduces to $\Gamma$, which in turn simplifies the one-step cost to $\psi_{\lambda}(\Gamma,Y,u)=c_{\lambda}^{0}(\Gamma,M,Y,u)$. As a result, MDP \ref{mdp2} becomes fully equivalent to MDP \ref{mdp3} under the MAF scheduling policy when $\rho=0$. On the other hand, it has been established in \eqref{c} that the MDP \ref{mdp1} is a classical sum-AoI optimization problem when $\rho=0$. Combining these two observations leads to conclusion (2) in Fig.~\ref{fig:simulation2}.
		
	    In contrast, for $\rho>0$, the solution to MDP \ref{mdp1} remains identical to that of MDP 3 under the MAF scheduling policy, as indicated by conclusion (1) in Fig.~\ref{fig:simulation2}. These structural equivalences lead to the following corollary.
		
		\begin{corollary}[\textit{Age-Optimal Sampling is MSE-Optimal in the Infinite Horizon}]\label{coro1}
			
			Given the MAF scheduling, the MSE-optimal sampling is asymptotically equivalent to AoI-optimal sampling in the infinite horizon.
		\end{corollary}
		This result allows us to directly leverage age-optimal sampling policies in our estimation framework. In particular, we can adopt the low-complexity threshold-triggered water-filling sampling (WF) policy in~\cite[Eq. 44]{9358178}, given by:
		\begin{equation}\label{thr}
			Z_i=\max\left\{0,T-\frac{\Delta_1(D_{i})+\Delta_2(D_{i})}{2}\right\}, \forall i\in\mathbb{N},
		\end{equation}
		where the threshold $T$ can be obtained by golden-section methods \cite{press1992golden}. This policy is shown to achieve near-optimal sum-AoI performance for multiple-sources scenarios.
									
				\section{Simulation Results}\label{sectionVIII}
				We evaluate the performance of the proposed joint sampling and scheduling policy under correlated Wiener processes. Unless otherwise specified, the correlation level is fixed at \( \rho = 0.9 \), and MMSE fusion estimation is implemented at the receiver. As a case study, the transmission delay is modeled as a binary random variable: 0 with probability \( p \), and \( Y_{\max} \) with probability \( 1 - p \).
				
				\vspace{0.5em}
				\textbf{Benchmark Policies:}
				\begin{itemize}
					\item \textbf{Zero-Wait (ZW) Sampling:} Each source samples right after the previous transmission completes, i.e., \( Z_i = 0 \).
					\item \textbf{Constant-Wait Sampling:} A fixed delay \( Z_i = d \) is imposed between transmissions.
					\item \textbf{Random (RAND) Scheduling:} Each sensor is selected at random, i.e., \( \Pr(a_i = m) = \frac{1}{2} \) for all \( m \in \{1,2\} \).
				\end{itemize}
				Our proposed MAF scheduling \eqref{eq33} plus Water-filling sampling \eqref{thr} are jointly referred to by WF+MAF.
				
				\begin{figure}
					\centering
					\includegraphics[width=0.65\linewidth]{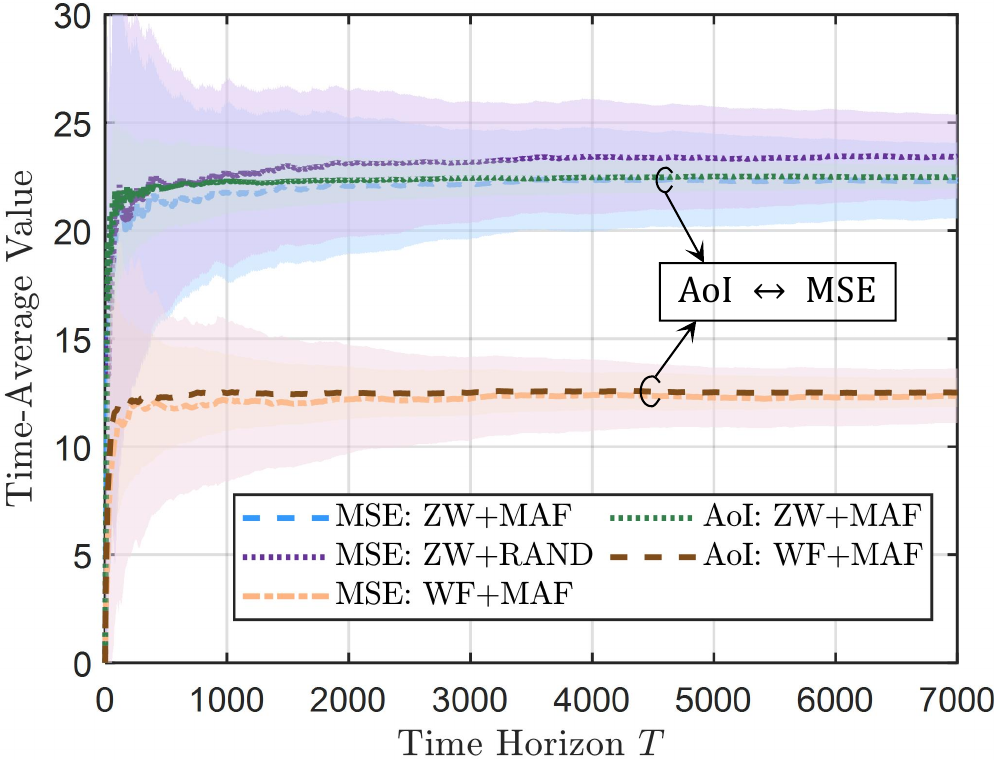}
					\caption{Finite-horizon time-average MSE and AoI under different scheduling and sampling policies. Here $Y_{\max}=20$ and $p=0.95$.}
					\label{fig:simulation1}
				\end{figure}
				
				\begin{figure}[htbp]
					\centering
					\begin{minipage}[t]{0.5\linewidth}
						\centering
						\includegraphics[width=\linewidth]{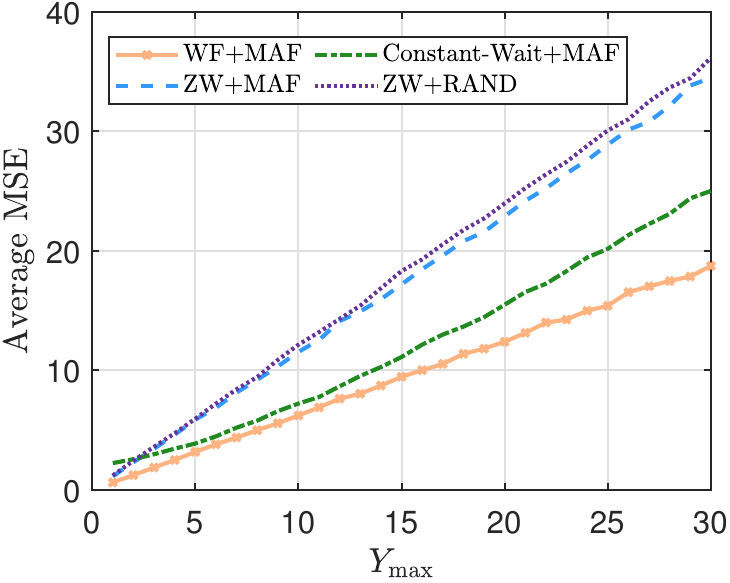}
						\textbf{(a)} MSE vs. $Y_{\max}$
						\label{fig:minipage1}
					\end{minipage}
					\hfill
					\begin{minipage}[t]{0.48\linewidth}
						\centering
						\includegraphics[width=\linewidth]{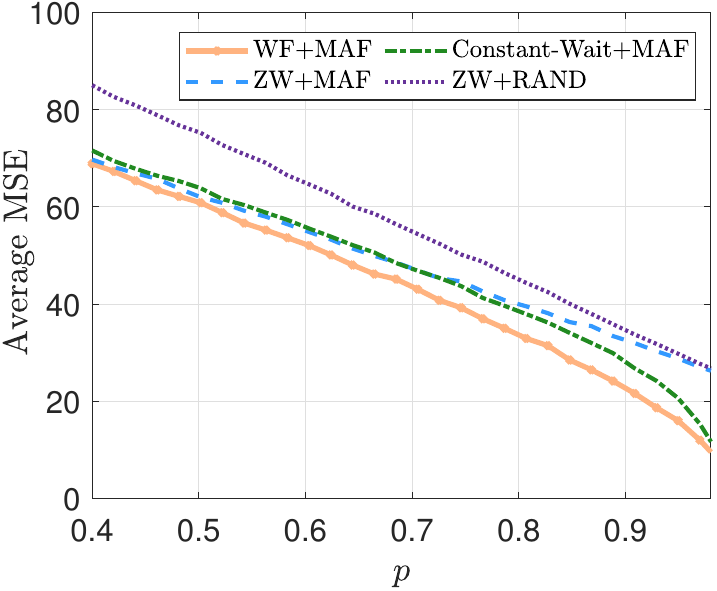}
						\textbf{(b)} MSE vs. $p$.
						\label{fig:minipage2}
					\end{minipage}
					\caption{Long-term average MSE and AoI under different scheduling and sampling policies. We set $d=1$ for constant-wait sampling. (a) $p=0.95$; (b) $Y_{\max}=25$.}
					\label{fig:twominipage}
				\end{figure}			
				
				Fig.~\ref{fig:simulation1} shows that time-average MSE and AoI converge as the horizon \( T \) increases, validating the \textit{asymptotic equivalence} in Corollary~\ref{coro1}. Notably, AoI serves as a provable upper bound for MSE, as implied by~\eqref{Delta}, and the gap between the two metrics vanishes as $T\to\infty$. 
				
				Fig.~\ref{fig:twominipage} compares long-term MSE across policy combinations. WF+MAF consistently achieves the lowest MSE. In Fig.~(a), its performance degrades slowly with increasing \( Y_{\max} \), unlike ZW-based strategies which exhibit sharp deterioration. In Fig.~(b), WF+MAF remains robust across the parameter $p$. These results demonstrate that the proposed joint sampling and scheduling policy offers strong performance and robustness under various channel delay distributions.

				\section{Conclusion}\label{sectionVI}
				In this paper, we studied a remote information fusion problem in which a remote fusion center aims to track a multi-dimensional \textit{correlated} Wiener process based on status updates from multiple distributed sources communicating over a shared channel with random delays. We explicitly solved the joint optimization of sampling, scheduling, and estimation policies, and derived a structural characterization of the optimal solution that is both theoretically optimal and practically implementable. Notably, we established an asymptotic equivalence between the time-averaged AoI and MSE under the optimal MAF scheduling and MMSE estimation, even in the presence of inter-source correlation. Future work includes extending our framework to larger-scale systems with more sensors, and exploring other classes of Markovian sources beyond the correlated Wiener process.
				
				\bibliographystyle{IEEEtran}
				\bibliography{reference}

				\appendices
				\normalsize
				\section{Proof of Theorem \ref{the1}}\label{proof:the1}
				Consider the following zero-mean Gaussian vector:
				\begin{equation}
					\mathbf{Z}\triangleq\bigl(X,\underbrace{Y_1,Y_2}_{Y}\bigr)^{\top}
					=\bigl(W_t^{(1)},\,W_{t_1}^{(1)},\,W_{t_2}^{(2)}\bigr)^{\top}.
				\end{equation}
				Its block covariance matrix, obtained from the coupled Wiener process construction, is:
				\begin{equation}
				{\Sigma}=\begin{pmatrix}
					\Sigma_{XX} & \Sigma_{XY}\\
					\Sigma_{YX} & \Sigma_{YY}
				\end{pmatrix}
				\!=\!
				\begin{pmatrix}
					t &\!\! [\,t_1\; \rho t_2]\\[4pt]
					[\,t_1\; \rho t_2]^{\!\top} &
					\begin{psmallmatrix}
						t_1 & \rho t_{\min}\\
						\rho t_{\min} & t_2
					\end{psmallmatrix}
				\end{pmatrix},
			\end{equation}
				where $t_{\min}=\min\{t_1,t_2\}$.  
				Because $X,Y_1,Y_2$ are jointly Gaussian random variables, one can calculate the following conditional expectation \cite[Chap. 2]{anderson1958introduction}:
				\begin{equation}\label{eq46}
					\mathbb{E}[X\mid Y_1,Y_2]=\Sigma_{XY}\,\Sigma_{YY}^{-1}
					\begin{pmatrix}Y_1\\Y_2\end{pmatrix}.
				\end{equation}

				A direct inversion of the matrix $\Sigma_{YY}$ yields:
				\begin{equation}
					\Sigma_{YY}^{-1}=\frac{1}{\Delta}
					\begin{pmatrix}
						t_2 & -\rho t_{\min}\\
						-\rho t_{\min} & t_1
					\end{pmatrix},
					\quad
					\Delta=t_1t_2-\rho^{2}t_{\min}^{2},
				\end{equation}
				and therefore
				\begin{equation}
				\Sigma_{XY}\Sigma_{YY}^{-1}
				=\frac1{\Delta}\,[\,t_2(t_1-\rho^{2}t_{\min}),\; \rho t_1(t_2-t_{\min})\,].
			\end{equation}
				Substituting $X=W_{t}^{(1)}$, $Y_1=W_{t_1}^{(1)}$ and $Y_2=W_{t_2}^{(2)}$ into \eqref{eq46} gives the conditional expectation.  
				Setting $t_1=t-\Delta_1(t)$ and $t_2=t-\Delta_2(t)$ immediately yields~\eqref{eq11}. Symmetrically, we can also establish \eqref{eq12}.
				\section{Proof of \eqref{Gammacompare}}\label{appendixa}
				 For $k\leq \ell$, since $a_k = a_k'$ for all $k \leq \ell$, by the update rule~\eqref{eq:gap-env-transition}, we have $\Gamma_k = \Gamma_k'$ for these steps.
				 
				 For $k = \ell+1$, at step $\ell+1$, the actions differ, so
				\begin{equation}
					\Gamma_{\ell+1}=\Gamma_\ell+Y_\ell+Z_\ell,\qquad \Gamma_{\ell+1}'=Y_\ell+Z_\ell.
				\end{equation}
				which gives $\Gamma_{\ell+1} > \Gamma_{\ell+1}'$ since $\Gamma_\ell > 0$.

				For $k \geq \ell+2$, We use induction. Suppose for some $n \geq \ell+1$ we have $\Gamma_n \geq \Gamma_n'$.
				Since $a_n = a_n'$ for $n > \ell$, we have:
				\begin{equation}
					\Gamma_{n+1}-\Gamma_{n+1}'=\begin{cases}
						\Gamma_n -\Gamma_n',
						& \text{if } a_n=\mathop{\arg\max}_{j}\{S_{n,j}\},\\[4pt]
						0
						& \text{if } a_n=\mathop{\arg\min}_{j}\{S_{n,j}\}.
					\end{cases}
				\end{equation}
				Both cases preserve the inequality $\Gamma_{n+1} \geq \Gamma_{n+1}'$. Thus, by induction, $\Gamma_k \geq \Gamma_k'$ for all $k$. \hfill$\blacksquare$	
				\section{Proof of Theorem \ref{lemma3}}\label{proof:the3}
				From the update rule \eqref{Mdynamics}, we have 
				\begin{equation}
					M_n=M_0+\sum_{i=1}^{n}(Y_{i-1}+Z_{i-1}).
				\end{equation}
				By the \textit{Kolmogorov strong large number law}, 
				\begin{equation}
					\lim_{n\to\infty}\frac{M_n}{n}>\lim_{n\to\infty}\frac{M_0}{n}+\frac{\sum_{i=1}^{n}Y_{i-1}}{n}\mathop{\rightarrow}_{\text{a.s.}}\mu_Y.
				\end{equation}
				Therefore, there exists $k_0$ such that for all $i \geq k_0$,
				\begin{equation}\label{labela}
					M_i>\frac{\mu_Y}{2}i, \qquad\forall i\ge k_0.
				\end{equation}
				Denote $y^{(u)}$ and $z^{(u)}$ as the maximum value of $Y_i$ and $Z_i$, respectively. From \eqref{dynamics2Gamma}, we obtain that:
				\begin{equation}\label{labelb}
					\Gamma_i^2=(Y_{i-1}+Z_{i-1})^2\le 2 (y^{(u)})^2+2(z^{(u)})^2.
				\end{equation}
				Substitute the above bounds into the left-hand side of \eqref{equivalence}:
					\begin{equation}
						\begin{aligned}
							&\lim_{n\to\infty}
							\frac{1}{n}\sum_{i=0}^{n-1}
							\mathbb{E}[\psi_\lambda(\Gamma_i,Y_i,u_i)-c_\lambda^{\rho}(\Gamma_i,M_i,Y_i,u_i)]\\
							&=\lim_{n\to\infty}
							\frac{1}{n}\sum_{i=0}^{n-1}(Z_i+\mu_Y)\left(\frac{\rho^{2}\Gamma_i^2}{\,\Gamma_i+(1-\rho^{2})M_i\,}\right)\\
							&<\rho^2(z^{(u)}+\mu_Y)\lim_{n\to\infty}
							\frac{1}{n}\sum_{i=1}^{k_0-1}\frac{\Gamma_i^2}{(1-\rho^2)M_i}+\sum_{i=k_0}^{n}\frac{\Gamma_i^2}{(1-\rho^2)M_i}\\
							&<\frac{2\rho^2(z^{(u)}+\mu_Y)}{\mu_Y(1-\rho^2)}\lim_{n\to\infty}
							\frac{1}{n}\sum_{i=k_0}^{n}\frac{\Gamma_i^2}{i}\\
							&\le\frac{4\rho^2(z^{(u)}+\mu_Y)((y^{(u)})^2+(z^{(u)})^2)}{\mu_Y(1-\rho^2)}\lim_{n\to\infty}
							\frac{1}{n}\sum_{i=k_0}^{n}\frac{1}{i}=0,
						\end{aligned}
					\end{equation}
					where the last equality follows since $\frac{1}{n} \sum_{i=k_0}^n \frac{1}{i} \to 0$ as $n \to \infty$.
					By the squeeze theorem, we complete the proof.\hfill$\blacksquare$	
			\end{document}